\documentclass{amsart}
\usepackage{amsfonts}

\begin{document}

\vfuzz2pt 

 \newtheorem{thm}{Theorem}[subsection]
\newtheorem{cor}[thm]{Corollary}
 \newtheorem{lem}[subsection]{Lemma}
 \newtheorem{prop}[subsection]{Proposition}
 \theoremstyle{definition}
 \newtheorem{defn}[subsection]{Definition}
 \theoremstyle{remark}
 \newtheorem{rem}[subsection]{Remark}
 \numberwithin{equation}{subsection}
\newcommand{\CC}{\mathbb{C}}
\newcommand{\KK}{\mathbb{K}}
\newcommand{\ZZ}{\mathbb{Z}}
\newcommand{\NN}{\mathbb{N}}
\def\a{{\alpha}}

\def\b{{\beta}}

\def\d{{\delta}}

\def\g{{\gamma}}

\def\l{{\lambda}}

\def\gg{{\mathfrak g}}
\def\cal{\mathcal }

\title{Filiform Lie algebras of order $3$}

\author{R.M. Navarro}

\address{Rosa Mar{\'\i}a Navarro.\newline \indent
Dpto. de Matem{\'a}ticas, Universidad de Extremadura, C{\'a}ceres
(Spain)}

\email{rnavarro@unex.es}



\begin{abstract}

The aim of this work is to generalize a very important type of Lie algebras and superalgebras,
i.e. filiform Lie (super)algebras, into the theory of Lie algebras of order $F$.
Thus, the concept of filiform Lie algebras of order $F$ is obtained. In particular, for $F=3$ it has been proved that by using infinitesimal deformations of the associated model elementary Lie algebra it can be obtained families of filiform elementary lie algebras of order $3$,
 analogously as that occurs into the theory of Lie algebras (Vergne, $1970$).  Also we give the dimension, using an adaptation of the
$\mathfrak{sl}(2,\CC)$-{\it module Method}, and a basis of such
infinitesimal deformations in some generic cases.

 \end{abstract}

\maketitle

\section{Introduction}

The identification and classification of Lie algebras have had
important applications to the study of symmetries in physics.
Nowadays such symmetries are not limited to the geometrical ones
of space-time, arising thus the concept of supersymmetry and
consequently Lie superalgebra. Among others, the possible
generalizations of Lie superalgebras that have been proven to be
physically relevant are color Lie superalgebras
(\cite{filiformcolor}, \cite{dimension_color}, \cite{erratum2},
\cite{5}) and Lie algebras of order $F$ (\cite{traubenberg2000},
\cite{traubenberg2002}, \cite{traubenberg2007}).

\

 In this paper we shall consider Lie algebras of order $F$, that constitute the
 underlying algebraic structure associated to fractional supersymmetry
  (\cite{A1}, \cite{A2}, \cite{20}, \cite{21}), (note that a
different point of view can be seen in \cite{kerner4}). Thus, Lie
algebras of order $3$ (or more generally Lie algebras of order
$F$) were introduced as a possible generalisation of Lie
superalgebras, in order to implement non-trivial extensions of the
Poincar\'e symmetry which are different than the usual
supersymmetric extension. In particular, a Lie algebra of order
$F$ admits a $\ZZ_F$-grading, $\gg=\gg_0\oplus\gg_1\oplus\gg_2
\cdots \oplus \gg_{F-1}$ with the zero-graded part being a complex
Lie algebra, and it also admits an $F$-fold symmetric product $
\cal{S}^F(\gg_i)$, $1 \leq i \leq F-1$.

\

In this paper our goal is to generalize the concept of filiform
Lie (super)algebras into the theory of Lie algebras of order $F$.
We obtain thus, the notion of filiform Lie algebras of order $F$.

\

The concept of filiform Lie algebras was firstly introduced in
\cite{Vergne} by Vergne. This type of nilpotent Lie algebra has
important properties; in particular, every filiform Lie algebra
can be obtained by a deformation of the model filiform algebra
$L_n$. In the same way as filiform Lie algebras, all filiform Lie
superalgebras can be obtained by infinitesimal deformations of the
model Lie superalgebra $L^{n,m}$ \cite{Bor07}, \cite{JGP2} and
\cite{JGP4}.

\

In this paper we generalize this concept obtaining {\it filiform
Lie algebras of order $F$} and the {\it model filiform Lie algebra
of order $3$}. We have proved that by using infinitesimal deformations of the associated model elementary Lie algebra it can be obtained families of filiform elementary lie algebras of order $3$ (see Theorem 2).

\

For some generic cases (see Theorems $3$ and $4$) we have given
the dimension and a basis of the mentioned infinitesimal
deformations. For to do that we have used an adaptation of the
$\mathfrak{sl}(2,\CC)$-{\it module Method}. Also, we have given
some properties of the algebraic variety of elementary Lie
algebras of order $3$.

\

We do assume that the reader is familiar
with the standard theory of Lie algebras. All the vector spaces
that appear in this paper (and thus, all the algebras) are assumed
to be  ${\mathbb C}$-vector spaces with finite dimension.

\section{Preliminaries}

The vector space $V$ is said to be
$\ZZ_F-$graded if it admits a decomposition in direct sum, $V=V_0
\oplus V_1 \oplus \cdots V_{F-1}$, with $F \in \NN^{*}$. An element $X$ of $V$ is called
homogeneous of degree $\g$ ($deg(X)=d(X)=\g$), $\g \in \ZZ_F$, if
it is an element of $V_{\g}$.

\

Let $V=V_0 \oplus V_1 \oplus \cdots V_{F-1}$ and $W=W_0\oplus W_1
\oplus \cdots W_{F-1}$ be two graded vector spaces. A linear
mapping $f: V \longrightarrow W$ is said to be homogeneous of
degree $\g$ ($deg(f)=d(f)=\g$), $\g \in \ZZ_F$, if $ f(V_{\a})
\subset W_{\a + \g (mod \ F)}$ for all $\a \in \ZZ_F$. The mapping
$f$ is called a homomorphism of the $\ZZ_F-$graded vector space
$V$ into the $\ZZ_F-$graded vector space $W$ if $f$ is homogeneous
of degree 0. Now it is evident how we define an isomorphism or an
automorphism of $\ZZ_F-$graded vector spaces.

\

A superalgebra $\gg$ is just a $\ZZ_2-$graded algebra $\gg=\gg_0
\oplus \gg_1$,  \cite{Kac} and \cite{Scheunert}. That is, if we denote by $[\  , \
]$ the bracket product of $\gg$, we have $[\gg_{\alpha},
\gg_{\beta}]\subset \gg_{\alpha+\beta (mod 2)}$ for all $\alpha,
\beta \in \ZZ_2$.

\

\begin{defn} \rm \label{defA} Let $\gg=\gg_0\oplus\gg_1$
be a superalgebra whose multiplication is denoted by the bracket
product [ , ]. We call $\gg$ a {\bf Lie superalgebra} if the
multiplication satisfies the following identities:

 1. $[X,Y]=-(-1)^{\alpha \cdot \beta}[Y,X]\qquad \forall X\in \gg_{\alpha},
  \forall Y \in \gg_{\beta}$.

 2. $(-1)^{\g \cdot \alpha}[X,[Y,Z]]+(-1)^{\alpha \cdot \beta}[Y,[Z,X]]+
 (-1)^{\beta \cdot \g}[Z,[X,Y]]=0$
 \newline \indent \qquad for all $X\in \gg_{\alpha}, Y \in
 \gg_{\beta}, Z
 \in \gg_{\g}$ with  $\a, \b, \g  \in \ZZ_2$.

 \noindent Identity 2 is called the graded Jacobi identity and it will
 be denoted by $J_g(X,Y,Z)$.
\end{defn}

We observe that if $\gg=\gg_0\oplus\gg_1$ is a Lie superalgebra,
we have that $\gg_0$ is a Lie algebra and $\gg_1$ has structure of
$\gg_0-$module.

\

Next we recall the definition and some basic properties of Lie
algebras of order $F$ introduced in \cite{traubenberg2000},
\cite{traubenberg2002} and \cite{traubenberg2007}.

\begin{defn} \cite{campoamor} Let $F \in \NN^{*}$. A $\ZZ_F$-graded $\CC$-vector
space $\gg=\gg_0\oplus\gg_1\oplus\gg_2 \cdots \oplus \gg_{F-1}$ is
called a complex {\bf Lie algebra of order $F$} if the following hold:

\begin{itemize}

\item[(1)] $\gg_0$ is a complex Lie algebra.

\item[(2)] For all $i=1, \dots,F-1$, $\gg_i$ is a representation
of $\gg_0$. If $X \in \gg_0$, $Y \in \gg_i$, then $[X,Y]$ denotes
the action of $X \in \gg_0$ on $Y \in \gg_i$ for all $i=1,
\dots,F-1$.

\item[(3)] For all $i=1, \dots,F-1$, there exists an $F$-Linear,
$\gg_0$-equivariant map, $ \{ \cdots \}: \cal{S}^F(\gg_i) \longrightarrow \gg_0,$
 where $\cal{S}^F(\gg_i)$ denotes the $F$-fold symmetric
product of $\gg_i$.

\item[(4)] For all $X_i \in \gg_0$ and $Y_j \in \gg_k$, the following ``Jacobi identities" hold:
\end{itemize}
\begin{equation}\label{1}[[X_1,X_2],X_3]+ [[X_2,X_3],X_1]+
[[X_3,X_1],X_2]=0.
\end{equation}
\begin{equation}\label{2}[[X_1,X_2],Y_3]+ [[X_2,Y_3],X_1]+
[[Y_3,X_1],X_2]=0.
\end{equation}
\begin{equation}\label{3}[X,\{ Y_1,\dots,Y_F\}]= \{[X,Y_1],
\dots,Y_F\}+ \dots + \{ Y_1, \dots,[X,Y_F]\}.
\end{equation}
\begin{equation}\label{4}\sum_{j=1}^{F+1}[Y_j, \{Y_1, \dots, Y_{j-1},Y_{j+1}, \dots, Y_{F+1}\}]=0.
\end{equation}
\end{defn}

\begin{rem} We observe that a Lie algebra of order $1$ it is just
a Lie algebra and a Lie algebra of order $2$ it is a Lie
superalgebra. Thus, Lie algebras of order $F$ can be seen as a
generalization of Lie algebras and superalgebras.
\end{rem}

\begin{prop} \cite{campoamor} Let $\gg=\gg_0 \oplus \gg_1 \oplus \cdots \oplus
\gg_{F-1}$ be a Lie algebra of order $F$, with $F>1$. For any
$i=1,\dots,F-1$, the subspaces $\gg_0 \oplus \gg_i$ inherits the
structure of a Lie algebra of order $F$. We call these type of
algebras {\bf elementary Lie algebras of order $F$}.
\end{prop}

We will restrict our study to
elementary Lie algebras of order $3$, $\gg=\gg_0 \oplus \gg_1$.
Examples of elementary Lie algebras of order $3$ can be seen in  \cite{traubenberg2007}.

\

\begin{defn}

A representation of an elementary Lie algebra of order $F$ is a
linear map $ \rho :\gg=\gg_0 \oplus \gg_1 \longrightarrow End(V)$,
such that for all $X_i \in \gg_0, \ Y_j \in \gg_1,$

\

$\begin{array}{rcl}
\rho([X_1,X_2]) & = & \rho(X_1)\rho(X_2)-\rho(X_2)\rho(X_1)\\
\rho([X_1,Y_2]) & = & \rho(X_1)\rho(Y_2)-\rho(Y_2)\rho(X_1)\\
\rho \{ Y_1,\dots,Y_F\} & = & \displaystyle\sum_{\sigma \in S_F}
\rho(Y_{\sigma(1)}) \cdots \rho(Y_{\sigma(F)})
\end{array}$

\

\noindent $S_F$ being the symmetric group of $F$ elements.

\end{defn}

By construction, the vector space $V$ is graded $V=V_0 \oplus
\cdots \oplus V_{F-1}$, and for all $a=\{ 0, \dots, F-1\}$, $V_a$
is a $\gg_0$-module. Further, the condition
$\rho(\gg_1)(V_a)\subseteq V_{a+1}$ holds.

\section{Filiform Lie algebras of order $F$}

In this section we will focus our study in generalize a very
important type of nilpotent Lie algebras, i.e. filiform Lie
algebras obtaining the notion of filiform Lie
algebras of order $F$.

To do that we will star with a previous concept ``{\it filiform
module}".

\begin{defn} Let $\gg=\gg_0 \oplus \gg_1 \oplus \cdots \oplus
\gg_{F-1}$ be a
 Lie algebra of order $F$. $\gg_i$ is called a {\bf $\gg_0$-filiform module} if
there exists a decreasing subsequence of vector subspaces in its
underlying vectorial space $V$, $V=V_m \supset \dots \supset V_1
\supset
 V_0$, with dimensions
$m,m-1,\dots 0$, respectively, $m>0$, and such that
$[\gg_0,V_{i+1}]= V_{i}$.
\end{defn}

\begin{defn} Let  $\gg=\gg_0 \oplus \gg_1 \oplus \cdots \oplus
\gg_{F-1}$ be a
 Lie algebra of order $F$. Then $\gg$ is a {\bf filiform  Lie
algebra of order $F$} if the following conditions hold:
\begin{itemize}
\item[(1)] $\gg_0$ is a filiform Lie algebra.

\item[(2)] $\gg_i$ has structure of $\gg_0$-filiform module, for all
$i, \ 1 \leq i \leq F-1$
\end{itemize}
\end{defn}

From now on we will restrict our study to $F=3$. Thus, if we take
an homogeneous basis $\{ X_0,\dots,X_{n},$ $Y_1, \dots,Y_m,$
$Z_1,\dots,Z_p\}$ of a Lie algebra of order $3$ $\gg=\gg_0\oplus
\gg_1 \oplus \gg_2$ with $X_i \in \gg_0$, $Y_j \in \gg_1$ and $Z_k
\in \gg_2$, then the Lie algebra of order $3$ will be completely
determined by its structure constants, that is, by the set of
constants $\{
C_{ij}^k,D_{ij}^k,E_{ij}^k,F_{ijl}^k,G_{ijl}^k\}_{i,j,l,k}$ that
verify

$$\left\{\begin{array}{ll} [X_i,X_j]=\displaystyle
\sum_{k=0}^{n}C_{ij}^k X_k, & 0\leq i< j \leq
n,\\[1mm]
[X_i,Y_j]=\displaystyle \sum_{k=1}^{m}D_{ij}^k Y_k, & 0\leq i \leq
n, 1 \leq j \leq m,\\[1mm]
[X_i,Z_j]=\displaystyle \sum_{k=1}^{p}E_{ij}^k Z_k, & 0\leq i \leq
n, 1 \leq j \leq p,\\[1mm]
\{Y_i,Y_j,Y_l\}=\displaystyle \sum_{k=0}^{n}F_{ijl}^k X_k, & 1\leq
i \leq
 j \leq l \leq m,\\[1mm]
\{Z_i,Z_j,Z_l\}=\displaystyle \sum_{k=0}^{n}G_{ijl}^k X_k, & 1\leq
i \leq
 j \leq l \leq p,\\[1mm]
\end{array}\right.$$

with
$$C_{ij}^k=-C_{ji}^k, \ D_{ij}^k=-D_{ji}^k, \ E_{ij}^k=-E_{ji}^k$$
$$F_{ijl}^k=F_{ilj}^k=F_{jil}^k=F_{lij}^k=F_{jli}^k=F_{lji}^k $$
$$G_{ijl}^k=G_{ilj}^k=G_{jil}^k=G_{lij}^k=G_{jli}^k=G_{lji}^k $$

 By the  Jacobi identity we would have some polynomial equations
that the structure constants have to verify. All these equations
give to the set of Lie algebras of order $3$, denoted by ${\cal
L}_{n,m,p}$, the structure of algebraic variety.

We denote by ${\cal F}_{n,m,p}$ the subset of ${\cal L}_{n,m,p}$
composed of all filiform  Lie algebras of order $3$.

\

If we consider $A=(\gg_0 \wedge \gg_0) \oplus (\gg_0 \wedge \gg_1) \oplus (\gg_0 \wedge \gg_2) \oplus S^3(\gg_1) \oplus S^3(\gg_2)$
and the multiplication of the Lie algebra of order 3 $\gg=\gg_0 \oplus \gg_1 \oplus \gg_2$ as the linear map

$$\psi=(\psi_1,\psi_2,\psi_3,\psi_4,\psi_5): A \longrightarrow \gg \quad \mbox{where},$$

$\psi_1:\gg_0 \wedge \gg_0 \longrightarrow \gg_0,\quad$
$\psi_2:\gg_0 \wedge \gg_1 \longrightarrow \gg_1,\quad$
$\psi_3:\gg_0 \wedge \gg_2 \longrightarrow \gg_2,$

\

$\psi_4: S^3(\gg_1)\longrightarrow \gg_0,$ and $\psi_5: S^3(\gg_2)\longrightarrow \gg_0.$

\

Usually $\psi_1$, $\psi_2$ and $\psi_3$ are represented
by $[,]$, and $\psi_4$ and $\psi_5$ by $\{ , , \}$. If we consider the action of the group $GL(n+1,m,p)\cong GL(n+1)
\times GL(m) \times GL(p)$ on ${\cal L}_{n,m,p}$ we would have the following action with $(f_0,f_1,f_2) \in GL(n+1,m,p)$
$$(f_0,f_1,f_2)(\psi_1,\psi_2,\psi_3,\psi_4,\psi_5)\longrightarrow (\psi'_1,\psi'_2,\psi'_3,\psi'_4,\psi'_5),$$
where

$\psi'_1(X_1,X_2)=f_0^{-1} \psi_1(f_0(X_1),f_0(X_2)),\ $
$\psi'_2(X_1,Y_2)=f_1^{-1} \psi_2(f_0(X_1),f_1(Y_2)),$

\

$\psi'_3(X_1,Z_2)=f_2^{-1} \psi_3(f_0(X_1),f_2(Z_2)), \ $
$\psi'_4(Y_1,Y_2,Y_3)=f_0^{-1} \psi_4(f_1(Y_1),f_1(Y_2),f_1(Y_3)),$

\

and $\psi'_5(Z_1,Z_2,Z_3)=f_0^{-1} \psi_5(f_2(Z_1),f_2(Z_2),f_2(Z_3)).$

\

The group $GL(n+1,m,p)$ can be embedded in $GL(n+1+m+p)$ and it can be seen as the subgroup of $GL(n+1+m+p)$ which let the
subspaces $\gg_0$, $\gg_1$ and $\gg_2$ invariant. If we denote by ${\cal
O}_{\psi}$ the orbit of
$\psi=(\psi_1,\psi_2,\psi_3,\psi_4,\psi_5)$ with
respect to this action, then the algebraic variety ${\cal
L}_{n,m,p}$ is fibered by theses orbits. The quotient set is the
set of isomorphism classes of $(n+1+m+p)$-dimensional Lie algebras
of order $3$.

\

Prior to studying general classes of Lie algebras of order $3$ it
is convenient to solve the problem of finding a suitable basis; a
so-called adapted basis. This question is not trivial for Lie
algebras of order $3$ and it is very difficult to prove the
general existence of such a basis. However, for the class of
filiform Lie algebras of order 3, analogously as for ``filiform
color Lie superalgebras" \cite{filiformcolor}, it can be obtained
that there always exists an adapted basis. Thus we have the
following result.

\

\noindent{\bf Theorem 1. (Adapted basis)}  \label{base adaptada}
{\it Let $\gg=\gg_0\oplus \gg_1 \oplus \gg_2$ be a
 Lie algebra of order $3$. If $\gg$ is a {\bf filiform} Lie algebra of order $3$,
then there exists an adapted basis of $\gg$, namely $\{
X_0,\dots,X_{n}, Y_1, \dots,Y_m, Z_1,\dots,Z_p\}$ with $\{X_0,
X_1, \dots,$ $X_{n}\}$ a basis of $\gg_0$, $\{Y_1, \dots,Y_m\}$ a
basis of $\gg_1$ and $\{ Z_1,\dots,Z_p\}$ a basis of $\gg_2$, such
that:
$$\left\{\begin{array}{ll}
[X_0,X_i]=X_{i+1},& 1\leq i \leq n-1,\\[1mm]
 [X_0,Y_j]=Y_{j+1},& 1\leq
j
\leq m-1,\\[1mm]
[X_0,Z_k]=Z_{k+1},& 1\leq k
\leq p-1,\\[1mm]
[X_i,X_j]=\displaystyle \sum_{k=0}^{n}C_{ij}^k X_k, & 1\leq i< j
\leq
n,\\[1mm]
[X_i,Y_j]=\displaystyle \sum_{k=1}^{m}D_{ij}^k Y_k, & 1\leq i \leq
n, 1 \leq j \leq m,\\[1mm]
[X_i,Z_j]=\displaystyle \sum_{k=1}^{p}E_{ij}^k Z_k, & 1\leq i \leq
n, 1 \leq j \leq p,\\[1mm]
\{Y_i,Y_j,Y_l\}=\displaystyle \sum_{k=0}^{n}F_{ijl}^k X_k, & 1\leq
i \leq
 j \leq l \leq m,\\[1mm]
\{Z_i,Z_j,Z_l\}=\displaystyle \sum_{k=0}^{n}G_{ijl}^k X_k, & 1\leq
i \leq
 j \leq l \leq p,\\[1mm]
\end{array}\right.$$
$X_0$ will be called the characteristic vector.}

\

\begin{rem} We observe that for a $\gg_0$-filiform module there
exists a decreasing subsequence of vector subspaces in its
underlying vector space $V$, $V=V_m \supset \dots \supset V_1
\supset
 V_0$, with dimensions
$m,m-1,\dots 0$, respectively,  and such that $[\gg_0,V_{i+1}]=
V_{i}$. Thus, as $\gg_1$ is a $\gg_0$-filiform module this
decreasing subsequence of vector subspaces will be
$$<Y_1,\dots,Y_m>\supset <Y_2,\dots,Y_m> \supset \dots \supset <Y_m> \supset 0$$
and for $\gg_2$: $$<Z_1,\dots,Z_m>\supset <Z_2,\dots,Z_m> \supset
\dots \supset <Z_m> \supset 0$$
\end{rem}

\begin{rem} An adapted basis is composed by  homogeneous elements.
\end{rem}
\

\begin{defn} The {\bf model filiform  Lie algebra of order $3$}, called $\mu_0$,  is
the simplest filiform Lie algebra of order 3. It will be defined
in an adapted basis  $\{X_0, X_1, \dots, X_{n},$ $ Y_1, \dots,
Y_m,$ $Z_1, \dots,Z_p \}$ by the following non-null bracket
products $$\left\{\begin{array}{ll}
[X_0,X_i]=X_{i+1}, & 1 \leq i \leq n-1\\
\\[1mm]
[X_0,Y_j]=Y_{j+1},& 1 \leq j \leq m-1\\
\\ [1mm][X_0,Z_k]=Z_{k+1}& 1 \leq k \leq p-1
\end{array}\right.
$$

We observe that all the structure constants
$C_{ij}^k,D_{ij}^k,E_{ij}^k,F_{ijl}^k$ and $G_{ijl}^k$ that appear
in the theorem of adapted basis are all of them equal to zero.

\end{defn}

 \noindent {\bf Examples.} Other examples of filiform  Lie
 algebras of order $3$ distinct from the model are easy to obtain. Thus,

\begin{itemize}
\item[$1.$] $\mu^{1}_{n,m,p}$ is a family of non model filiform
Lie algebras of order $3$ and it can be expressed in an adapted
basis $\{X_0, X_1, \dots, X_{n},$ $ Y_1, \dots, Y_m,$ $Z_1,
\dots,Z_p \}$ by the following non-null bracket and $3$-bracket
products

$$\mu^{1}_{n,m,p}=\left\{\begin{array}{ll} [X_0,X_i]=X_{i+1}, & 1 \leq
i \leq n-1
\\[1mm]
[X_0,Y_j]=Y_{j+1},& 1 \leq j \leq m-1
\\ [1mm][X_0,Z_k]=Z_{k+1}& 1 \leq k \leq p-1\\[1mm]
\{ Y_1,Y_1,Y_1 \}=X_{n} & \\[1mm]
\{Z_1,Z_1,Z_1 \}=X_{n} & \\[1mm]
\end{array}\right.
$$

\item[$2.$] If $m\geq n$ then we have all a family of non model
filiform Lie algebras of order $3$: $\mu^{2}_{n,m,p}$ that can be
expressed in an adapted basis $\{X_0, X_1, \dots, X_{n},$ $ Y_1,
\dots, Y_m,$ $Z_1, \dots,Z_p \}$ by the following non-null bracket
and $3$-bracket products

$$\mu^{2}_{n,m,p}=\left\{\begin{array}{ll} [X_0,X_i]=X_{i+1}, & 1 \leq
i \leq n-1
\\[1mm]
[X_0,Y_j]=Y_{j+1},& 1 \leq j \leq m-1
\\ [1mm][X_0,Z_k]=Z_{k+1}& 1 \leq k \leq p-1\\[1mm]
\{ Y_1,Y_1,Y_1 \}=3X_{1} & \\[1mm]
\{ Y_1,Y_1,Y_2 \}=X_{2} & \\[1mm]
\{ Y_1,Y_1,Y_3 \}=X_{3} & \\[1mm]
\{ Y_1,Y_1,Y_4 \}=X_{4} & \\[1mm]
\vdots & \\[1mm]
\{ Y_1,Y_1,Y_{n} \}=X_{n} & \\[1mm]
\end{array}\right.
$$
\end{itemize}

\

From now on we are going to restrict our study to elementary Lie algebras of order $3$ due to its physicals applications.

\section{The algebraic variety of elementary Lie algebras of order 3}

By the Jacobi identity we have some polynomial equations that the
structure constants hold. All these equations give to the set of
elementary Lie algebras of order 3 the structure of algebraic
variety (analogously as for Lie algebras of order 3).

We denote by ${\cal L}_{n,m}$ the mentioned algebraic variety and by ${\cal F}_{n,m}$ the subset of ${\cal L}_{n,m}$ composed of all filiform elementary Lie algebra of order 3.

Next, we are going to define some descending sequences of ideals.

\begin{defn} Let $\gg=\gg_0\oplus\gg_1$ be an elementary
Lie algebra of order $3$. Then, we define the descending sequences of
ideals ${\cal C}^{k}(\gg_0)$ and ${\cal C}^{k}(\gg_{1})$,
as follows:

 $${\cal C}^0(\gg_{0})=\gg_{0}, \quad {\cal C}^{k+1}(\gg_0)=[\gg_0, {\cal
 C}^k(\gg_0)], \quad
 k\geq 0$$
 and
 $${\cal C}^0(\gg_1)=\gg_1, \quad {\cal C}^{k+1}(\gg_1)=[\gg_0, {\cal
 C}^k(\gg_1)], \quad
 k\geq 0$$
\end{defn}

Using the descending sequences of ideals defined above we give an
invariant of elementary Lie algebras of order $3$ called {\bf order-nilindex}.

\begin{defn}
 If $\gg=\gg_0\oplus \gg_1$ is an elementary Lie algebra of order $3$, then $\gg$ has {\bf order-nilindex}
$(p_0,p_1)$, if the following conditions holds: $$ {\cal
C}^{p_0-1}(\gg_0), \ {\cal C}^{p_1-1}(\gg_1)\neq 0$$ and  $${\cal
C}^{p_0}(\gg_0)={\cal C}^{p_1}(\gg_1)=0$$
\end{defn}

\

\begin{rem} Note that are equivalent the definitions of ``filiform" elementary Lie algebras of order $3$ and ``maximal order-nilindex",
maximal in  the sense of  lexicographic order. That is, if
$\gg=\gg_0\oplus \gg_1$, with $dim(\gg_0)=n+1$ and $dim(\gg_1)=m$,
is a filiform elementary Lie algebra of order 3 then $\gg$ has
maximal order-nilindex $(n,m)$ and reciprocally.
\end{rem}

We note by ${\cal N}^{p_0, p_1}_{n,m}$ the subset of ${\cal L}_{n,m}$ composed of all the elementary Lie algebras of order $3$ with order-nilindex $(r,s)$ where $r \leq p_0$ and $s\leq p_1$.

\begin{prop} ${\cal N}^{p_0, p_1}_{n,m}$ is an algebraic subvariety of ${\cal L}_{n,m}$.
\end{prop}

\begin{proof} The set ${\cal N}^{p_0, p_1}_{n,m}$ of ${\cal L}_{n,m}$ is defined by the restrictions ${\cal C}^{p_0} (\gg_0)=0$ and ${\cal C}^{p_1} (\gg_1)=0$, but these restrictions are polynomial equations of the structure constants. Thus ${\cal N}^{p_0, p_1}_{n,m}$ is closed for the Zariski topology and it will have the structure of an algebraic subvariety. We denote by ${\cal N}^{p_0, p_1}_{n,m}$ the corresponding affine variety.
\end{proof}

For simplicity we will refer to ${\cal N}^{n, m}_{n,m}$ as ${\cal
N}_{n,m}$.

\begin{prop} Each component of ${\cal F}_{n,m}$ determines a component of ${\cal N}_{n,m}$.
\end{prop}

\begin{proof} ${\cal F}_{n,m}$=${\cal N}_{n,m} - {\cal N}^{n-1, m-1}_{n,m}$ is a Zariski open subset of ${\cal N}_{n,m}$.
\end{proof}

\begin{cor} For any $\gg \in {\cal F}_{n,m}$ the Zariski clousure of the orbit ${\cal O}(\gg)$, $\overline{{\cal O}(\gg)}^{\cal Z}$, will be an irreducible component of ${\cal N}_{n,m}$
\end{cor}
\section{The class of filiform Lie algebras of order $3$}

Recall that the concept of filiform Lie algebras was firstly
introduced in \cite{Vergne} by Vergne. This type of nilpotent Lie
algebra has important properties as it has been seen in section
above; in particular, every filiform Lie algebra can be obtained
by a deformation of the model filiform algebra $L_n$. In the same
way as filiform Lie algebras, all filiform Lie superalgebras can
be obtained by infinitesimal deformations of the model Lie
superalgebra $L^{n,m}$ \cite{Bor07}, \cite{JGP2} and \cite{JGP4}.
In this paper we generalize this result in part, for filiform Lie algebras
of order $3$. In particular, for elementary Lie algebras of order
$3$.

Next, we are going to generalize the concept of infinitesimal
deformations for elementary Lie algebras of order $3$. For more
details of  deformations of elementary Lie algebras of order $3$
see \cite{traubenberg2007}. Firstly we will start with a new
concept called pre-infinitesimal deformations.

\

\begin{defn} \label{def_deformation} Let $\gg=\gg_0\oplus \gg_1$ be an elementary
Lie algebra of order $3$ and let $A=(\gg_0 \wedge \gg_0) \oplus
(\gg_0 \wedge \gg_1) \oplus S^3(\gg_1)$. The linear map $\psi: A
\longrightarrow \gg$ is called a  {\bf pre-infinitesimal
deformation} of $\gg$ if it satisfies
$$\mu \circ \psi + \psi \circ \mu=0$$
with $\mu$
representing the law of $\gg$.

If we consider the restrictions of $\mu$ and $\psi$ to each of the
terms of $A$, i.e. $\psi=\psi_1 + \psi_2+\psi_3$ and
$\mu=\mu_1+\mu_2+\mu_3$ respectively, with
$$\psi_1, \ \mu_1: \gg_0 \wedge \gg_0 \longrightarrow \gg_0,$$
$$\psi_2, \ \mu_2: \gg_0 \wedge \gg_1 \longrightarrow \gg_1,$$
$$\psi_3, \ \mu_3: S^3(\gg_1) \longrightarrow \gg_0,$$
then the condition to be an infinitesimal deformation can be decomposed into $4$ equations:

\begin{itemize}
\item[$(1).$]  $\mu_1(\psi_1(X_i,X_j),X_k)+ \psi_1(\mu_1(X_i,X_j),X_k)+\mu_1(\psi_1(X_k,X_i),X_j)+$\\
$\psi_1(\mu_1(X_k,X_i),X_j)+\mu_1(\psi_1(X_j,X_k),X_i)+\psi_1(\mu_1(X_j,X_k),X_i)=0$ \\

\item[$(2).$]  $\mu_2(\psi_1(X_i,X_j),Y)+ \psi_2(\mu_1(X_i,X_j),Y)+\mu_2(\psi_2(X_j,Y),X_i)+$\\
$\psi_2(\mu_2(X_j,Y),X_i)+\mu_2(\psi_2(Y,X_i),X_j)+\psi_2(\mu_2(Y,X_i),X_j)=0$ \\

\item[$(3).$]  $\mu_1(X, \psi_3(Y_i,Y_j,Y_k))+ \psi_1(X,\mu_3(Y_i,Y_j,Y_k))-\mu_3(\psi_2(X_,Y_i),Y_j,Y_k)-$\\
$\psi_3(\mu_2(X_,Y_i),Y_j,Y_k)-\mu_3(Y_i,\psi_2(X_,Y_j),Y_k)-\psi_3(Y_i,\mu_2(X_,Y_j),Y_k)-$\\
$\mu_3(Y_i,Y_j,\psi_2(X_,Y_k))-\psi_3(Y_i,Y_j,\mu_2(X_,Y_k))=0$\\

\item[$(4).$]  $\mu_2(Y_i, \psi_3(Y_j,Y_k,Y_l))+ \psi_2(Y_i,\mu_3(Y_j,Y_k,Y_l))+\mu_2(Y_j,\psi_3(Y_i,Y_k,Y_l))+$\\
$\psi_2(Y_j,\mu_3(Y_i,Y_k,Y_l))+\mu_2(Y_k,\psi_3(Y_i,Y_j,Y_l))+\psi_2(Y_k,\mu_3(Y_i,Y_j,Y_l))+$\\
$\mu_2(Y_l,\psi_3(Y_i,Y_j,Y_k))+\psi_2(Y_l,\mu_3(Y_i,Y_j,Y_k))=0$\\

\end{itemize}
for all $X,X_i,X_j,X_k \in \gg_0$ and $Y,Y_i,Y_j,Y_k,Y_l \in \gg_1$.

\end{defn}

\

Next, we are going to present the definition of infinitesimal
deformations given in \cite{traubenberg2007}, page $21$.

\begin{defn} \cite{traubenberg2007} Let $\gg=\gg_0\oplus \gg_1$ be an elementary
Lie algebra of order $3$ and let $A=(\gg_0 \wedge \gg_0) \oplus
(\gg_0 \wedge \gg_1) \oplus S^3(\gg_1)$. The linear map $\psi: A
\longrightarrow \gg$ is called an  {\bf infinitesimal deformation}
of $\gg$ if it satisfies
$$\mu \circ \psi + \psi \circ \mu =0$$ and $$ \psi \circ \psi=0$$
with $\mu$ representing the law of $\gg$.

If we consider the restrictions of $\mu$ and $\psi$ to each of the
terms of $A$, i.e. $\psi=\psi_1 + \psi_2+\psi_3$ and
$\mu=\mu_1+\mu_2+\mu_3$ respectively, then the condition to be an
infinitesimal deformation can be decomposed into $8$ equations:
the four of being a pre-infinitesimal deformation and the four
equations that follow and that correspond to the condition $\psi
\circ \psi=0$

\begin{itemize}
\item[$(1).$] $\psi_1(\psi_1(X_i,X_j),X_k)+\psi_1(\psi_1(X_k,X_i),X_j)+\psi_1(\psi_1(X_j,X_k),X_i)=0$ \\

\item[$(2).$] $\psi_2(\psi_1(X_i,X_j),Y)+\psi_2(\psi_2(Y,X_i),X_j)+\psi_2(\psi_2(X_j,Y),X_i)=0$ \\

\item[$(3).$] $\psi_1(X,\psi_3(Y_i,Y_j,Y_k))-\psi_3(\psi_2(X,Y_i),Y_j,Y_k)-\psi_3(Y_i,\psi_2(X,Y_j),Y_k)-$\\
$\psi_3(Y_i,Y_j,\psi_2(X,Y_k))=0$\\

\item[$(4).$] $\psi_2(Y_i,\psi_3(Y_j,Y_k,Y_l))+\psi_2(Y_j,\psi_3(Y_i,Y_k,Y_l))+\psi_2(Y_k,\psi_3(Y_i,Y_j,Y_l))+$\\
$\psi_2(Y_l,\psi_3(Y_i,Y_j,Y_k))=0$\\

\end{itemize}
\end{defn}

\

\noindent{\bf Theorem 2.} {\it  If $\psi$ is a pre-infinitesimal
deformation of a model filiform elementary Lie algebra of order
$3$ law $\mu_0$ with $\psi(X_0,X)=0$ for all $X \in \mu_0$, then
the law $\mu_0 + \psi$ is a filiform Lie algebra of order $3$ law
iff $\psi$ is an infinitesimal deformation.}

\

\noindent {\bf Proof of the theorem.} Let $\psi$ be a
pre-infinitesimal deformation of $\mu_0$, with $\mu_0$ a model
filiform elementary Lie algebra of order $3$. We have too that
$\psi(X_0,X)=0 \ \forall X \in \mu_0$.

For obtaining the Jacobi identity, that is $\mu_0 \circ \mu_0 +
\mu_0 \circ \psi + \psi \circ \mu_0 + \psi \circ \psi=0$, as
$\mu_0$ represents a filiform elementary Lie algebra of order $3$
it verifies $\mu_0 \circ \mu_0 =0$ and as $\psi$ is a
pre-infinitesimal deformation we have $\mu_0 \circ \psi + \psi
\circ \mu_0=0$ and thus $\mu_0 + \psi$ verifies the Jacobi
identity iff $\psi \circ \psi=0$, i.e. $\psi$ is an infinitesimal
deformation. \hfill{\fbox{}}

\

Thus, by using infinitesimal deformations of the associated model elementary Lie algebra it can be obtained families of filiform elementary lie algebras of order $3$.

\

Recall that $\mu_0$ is the law of the model filiform elementary
Lie algebra of order $3$. Then, we denote by $Z(\mu_0)$ the vector
space composed by all the pre-infinitesimal deformations of
$\mu_0$, $\psi$, verifying that $\psi(X_0,X)=0, \ \forall \ X \in
\mu_0$. Next we will see that this vector space can be seen as
direct sum of three subspaces of pre-infinitesimal deformations
what facilities its study.

\begin{prop}Let $Z(\mu_0)$ be the vector space
composed by all the pre-infinitesimal deformations of $\mu_0$ that
vanish on the characteristic vector $X_0$. Then, if we note by
$\gg_0\oplus \gg_1$ the underlying vector space of $\mu_0$, i.e.
$\gg_0=< X_0, X_1, \dots, X_{n}>$ and $ \gg_1=<Y_1, \dots, Y_m>$,
we have that

\

 $\begin{array}{lll}Z(\mu_0) & = & Z(\mu_0) \cap Hom(\gg_0 \wedge \gg_0,\gg_0) \oplus Z(\mu_0) \cap Hom(\gg_0 \wedge \gg_1,\gg_1)
\\
  & & \oplus
 Z(\mu_0) \cap Hom(S^3(\gg_1),\gg_0)
 \\
 &:=& A \oplus B \oplus C \end{array}$
\end{prop}

\begin{proof} Let $\psi$ be such that $\psi \in Z(\mu_0)$. It is not difficult to
see that $\psi=\psi_1 + \psi_2 + \psi_3$ with $\psi_1 \in
Hom(\gg_0 \wedge \gg_0,\gg_0)$, $\psi_2 \in Hom(\gg_0 \wedge
\gg_1,\gg_1)$ and $\psi_3 \in Hom(S^3(\gg_1),\gg_0)$. In order to
complete the proof it only remains to verify that each of the
above homomorphisms is also a pre-infinitesimal deformation.

As $\psi=\psi_1+\psi_2+\psi_3$ is a pre-infinitesimal deformation it will verify the equations of Definition \ref{def_deformation}. Taking into account the law of $\mu_0$ and that $\psi(X_0,X)=0$, these equations remain as follows
\begin{itemize}
\item[$(1).$]  $ \psi_1(X_{j+1},X_k)+\psi_1(X_j,X_{k+1})+[\psi_1(X_j,X_k),X_0]=0$ \\

\item[$(2.1).$]  $[\psi_1(X_i,X_j),Y_k] =0, \qquad 1 \leq i,j$  \\

\item[$(2.2).$]  $\psi_2(X_{j+1},Y_k)+[\psi_2(X_j,Y_k),X_0]+\psi_2(X_j,Y_{k+1})=0$ \\

\item[$(3.1).$]  $[X_l, \psi_3(Y_i,Y_j,Y_k)]=0, \qquad 1 \leq l$\\

\item[$(3.2).$]  $[X_0, \psi_3(Y_i,Y_j,Y_k)]-\psi_3(Y_{i+1},Y_j,Y_k)-\psi_3(Y_i,Y_{j+1},Y_k)-\psi_3(Y_i,Y_j,Y_{k+1})=0$\\

\item[$(4).$]  $[Y_i, \psi_3(Y_j,Y_k,Y_l)]+ [Y_j,\psi_3(Y_i,Y_k,Y_l)]+[Y_k,\psi_3(Y_i,Y_j,Y_l)]+$\\ $[Y_l,\psi_3(Y_i,Y_j,Y_k)]=0$\\

\end{itemize}
for all $X_i,X_j,X_k, X_l \in \gg_0$ and $Y_i,Y_j,Y_k,Y_l \in
\gg_1$. From equations $(2.1)$ and $(3.1)$ we can obtain that $Im
\psi_1 \subset \langle X_1, \dots X_n \rangle$ and $Im \psi_3
\subset \langle X_1, \dots X_n \rangle$ respectively. In fact, the
ideal $\langle X_1, \dots X_n \rangle :=V_0$ is equal to its own
centralizer in $\gg_0$ and from equations $(2.1)$ and $(3.1)$ we
obtain that  $\psi_1(X_i,Y_j)$ and $\psi_3(Y_i,Y_j,Y_k)$
centralize $V_0$, and are thus elements of $V_0$. Then the
equation $(4)$ disappears, and each of the remaining equations
corresponds to the condition that has to verify each $\psi_i$ for
be a pre-infinitesimal deformation.
\end{proof}

\

\begin{rem} \label{/X_0} We note that if $\psi$ is a pre-infinitesimal
deformation of the model filiform elementary Lie algebra of order
$3$, $\psi=\psi_1+\psi_2+\psi_3 \in Z(\mu_0)$, then $Im \psi_3
\subset \langle X_1, \dots X_n \rangle$, that is
$$
   \psi_3: S^3(\gg_1)  \longrightarrow \gg_0/ \CC X_0
$$
such that
$$
 [X_0,\psi_3(Y_i,Y_j,Y_k)]- \psi_3
        ([X_0,Y_i],Y_j,Y_k)-\psi_3(Y_i,[X_0,Y_j],Y_k) - \psi_3(Y_i,Y_j,[X_0,Y_k]) =0
$$
\noindent with $1 \leq i \leq j \leq k \leq m$.

\end{rem}

\

Then, as the vector space of pre-infinitesimal deformations called
$Z(\mu_0)$ is equal to $A\oplus B \oplus C$ we will restrict our
study to each vector subspace. Of all of them, {\bf the most
important vector subspace will be $C$ because any
pre-infinitesimal deformation $\psi$ belonging to $C$ verifies
that $\psi \circ \psi=0$, i.e. $\psi$ is an infinitesimal
deformation}. Thus, $\mu_0+ \psi$ will be a filiform elementary
Lie algebra of order $3$ with $\psi \in C$.

\

\section{$\mathfrak{sl}(2,\CC)$-module Method}

In this section we are going to explain the $\mathfrak{
sl}(2,\CC)$-module method to compute the dimensions of $C$.

\

Recall the following well-known facts about the Lie algebra
$\mathfrak{sl}(2,\CC)$ and its finite-dimensional modules, see
e.g. \cite{Bourbaki7}, \cite{Humphreys}:

$\mathfrak{ sl}(2,\CC)=<X_{-},H,X_{+}>$ with the following commutation
relations:
$$
   \left\{\begin{array}{l}
      [X_{+},X_{-}]=H \\[1mm] [H,X_{+}]=2X_{+}, \\[1mm]
          [H,X_{-}]=-2X_{-}.
          \end{array}\right.
$$
Let $V$ be a $n$-dimensional  $\mathfrak{ sl}(2,\CC)$-module,
$V=<e_1,\dots,e_n>$. Then, up to isomorphism there exists a unique structure
of an
 irreducible $\mathfrak{ sl}(2,\CC)$-module in $V$ given in a basis
$e_1,\ldots,e_n$ as follows \cite{Bourbaki7}:
$$
   \left\{ \begin{array}{ll}
     X_{+}\cdot e_i=e_{i+1}, & 1 \leq i \leq n-1,\\[1mm]
            X_{+}\cdot e_n=0, &                  \\[1mm]
   H \cdot e_i=(-n+2i-1)e_i,  & 1 \leq i \leq n.
           \end{array}\right.
$$
It is easy to see that $e_n$ is the maximal vector of $V$ and its
weight, called the highest weight of $V$, is equal to $n-1$.

Let $V_0,V_1,\dots,V_k$ be $\mathfrak{ sl}(2,\CC)$-modules, then
the space $\mathrm{Hom}(\otimes_{i=1}^k V_i,V_0)$ is a
$\mathfrak{sl}(2,\CC)$-module in the following natural manner:
$$
  (\xi \cdot \varphi)(x_1,\dots,x_k)=\xi \cdot \varphi(x_1,\dots,x_k)-
   \sum_{i=1}^{i=k}
    \varphi (x_1,\dots,\xi \cdot x_i,x_{i+1},\dots,x_n)
$$
with $\xi \in \mathfrak{ sl}(2,\CC)$ and
$\varphi \in\mathrm{Hom}(\otimes_{i=1}^k V_i,V_0)$.

An element $\varphi \in\mathrm{Hom}(V_1\otimes V_1 \otimes V_1,V_0)$ is said to be invariant
if
   $X_{+}\cdot \varphi =0$,
that is
\begin{equation} \label{maximal}
      X_{+} \cdot \varphi(x_1,x_2,x_3)-
        \varphi (X_{+} \cdot x_1,x_2,x_3)- \varphi (x_1,X_{+} \cdot x_2,x_3)- \varphi (x_1,x_2,X_{+} \cdot x_3)=0,
         \ \forall x_1,x_2,x_3 \in V_1.
\end{equation}
Note that $\varphi \in\mathrm{Hom}( V_1\otimes V_1 \otimes
V_1,V_0)$ is invariant if and only if $\varphi$ is a maximal
vector.

On the other hand, we are going to consider the model filiform
elementary Lie algebra of order $3$ $\mu_0=\gg_0 \oplus \gg_1$
with basis $\{X_0,X_1,\dots,X_n,Y_1,\dots,Y_m\}$. By definition
(see Remark \ref{/X_0}) a pre-infinitesimal deformation $\varphi$
belonging to $C$ will be a symmetric multi-linear map:
$$
   \varphi: S^3(\gg_1)  \longrightarrow \gg_0/ \CC X_0
$$
such that
\begin{equation}\label{deformation}
 [X_0,\varphi(Y_i,Y_j,Y_k)]- \varphi
        ([X_0,Y_i],Y_j,Y_k)-\varphi(Y_i,[X_0,Y_j],Y_k) - \varphi(Y_i,Y_j,[X_0,Y_k]) =0
\end{equation}
\noindent with $1 \leq i \leq j \leq k \leq m$.

\

We are going to consider the structure of irreducible
$\mathfrak{sl}(2,\CC)$-module in $V_0=\langle
X_1,\ldots,X_n\rangle=\gg_0/\CC X_0$ and in $V_1=\langle
Y_1,\ldots,Y_m\rangle=\gg_1$, thus in particular:
$$
   \left\{ \begin{array}{ll}
      X_{+}\cdot X_i=X_{i+1}, & 1 \leq i \leq n-1,\\[1mm]
              X_{+}\cdot X_n=0,&                 \\ [1mm]
      X_{+}\cdot Y_j=Y_{j+1}, & 1 \leq j \leq m-1,\\[1mm]
              X_{+}\cdot Y_m=0.&
             \end{array}\right.
$$
We identify the multiplication of $X_{+}$ and $X_i$ in the
$\mathfrak{ sl}(2,\CC)$-module  $V_0=\langle X_1,\ldots,X_n\rangle$, with the
bracket $[X_0,X_i]$ in $\gg_0$. Analogously, we identify
$X_{+} \cdot Y_j$ and $[X_0,Y_j]$. Thanks to these
identifications, the expressions (\ref{maximal}) and (\ref{deformation})
are equivalent, so we have the following result:

\begin{prop} Any symmetric multi-linear map $\varphi$,
$\varphi: S^3 V_1 \longrightarrow V_0$ will be an
element of C if and only if $\varphi$ is a maximal vector of the
$\mathfrak{ sl}(2,\CC)$-module $\mathrm{Hom}(S^3 V_1,V_0)$, with
$V_0=\langle X_1,\ldots,X_n\rangle$ and $V_1=\langle Y_1,\ldots,Y_m\rangle$.
\end{prop}

\begin{cor} As each irreducible $\mathfrak{sl}(2,\CC)$-module has (up to nonzero scalar
multiples) a unique maximal vector, then the dimension of $C$ is
equal to the number of summands of any decomposition of
$\mathrm {Hom}(S^3 V_1,V_0)$ into the direct sum of irreducible
$\mathfrak{sl}(2,\CC)$-modules.
\end{cor}


Thanks to the symmetric structure of the weights, instead of to
sum the maximal vectors it is possible, and easier, to sum the
vectors of weight 0 or 1.

\begin{cor}\label{cor2} The dimension of C is equal to the dimension of the
subspace of  {\rm Hom} $(S^3 V_1,V_0)$ spanned by the vectors
of weight 0 or 1.
\end{cor}

\section{Computation of the
          dimension and a basis of $Z(\mu_0)\cap \mathrm{Hom}(S^3 \gg_1,\gg_0)$}

In this section we are going to apply the
$\mathfrak{sl}(2,\CC)$-module method above to $C=Z(\mu_0)\cap
\mathrm{Hom}(S^3 \gg_1,\gg_0)$.

Firstly, we consider a natural basis $\mathcal{B}$ of
$\mathrm{Hom}(S^3 \gg_1,\gg_0/ \CC X_0)$
consisting of the following maps where $1\leq s\leq n$ and $1\leq i,j,k,l,r,s
\leq m$:
$$
   \varphi_{i,j,k}^s(Y_l,Y_r,Y_s)=\left\{
         \begin{array}{ll}
           X_s & \mbox{ if }(i,j,k)=(l,r,s)\\
             0 & \mbox{ in all other cases}
         \end{array}
\right.
$$
Thanks to Corollary \ref{cor2} it will be enough to find the
basis vectors $\varphi_{i,j,k}^s$ with weight $0$ or $1$. The weight of
an element $\varphi_{i,j,k}^s$ (with respect to $H$) is
$$
  \lambda(\varphi_{i,j,k}^s)
    =\lambda(X_s)-\lambda(Y_i)-\lambda(Y_j)-\lambda(Y_k)=3m-n+2(s-i-j-k+1).
$$
In fact,

\noindent $\begin{array}{ll} (H \cdot
\varphi_{i,j,k}^s)(Y_i,Y_j,Y_k) & =H \cdot
\varphi_{i,j,k}^s(Y_i,Y_j,Y_k)-\varphi_{i,j,k}^s(H \cdot
Y_i,Y_j,Y_k)-\\ & \quad \varphi_{i,j,k}^s(Y_i,H \cdot Y_j, Y_k)-
\varphi_{i,j,k}^s(Y_i,Y_j,H \cdot Y_k) \\ & = H\cdot X_s
-\varphi_{i,j,k}^s((-m-1+2i)Y_i,Y_j,Y_k)-
\\ & \quad \varphi_{i,j,k}^s(Y_i,(-m-1+2j)Y_j,Y_k)-\\ & \quad
 \varphi_{i,j,k}^s(Y_i,Y_j,(-m-1+2k)Y_k) \\ & = (-n-1+2s)X_s-(-m-1+2i)X_s-\\ & \quad(-m-1+2j)X_s -
 (-m-1+2k)X_s\\ & = [3m-n+2(s-i-j-k+1)]X_s
    \end{array}
$


\begin{rem} If $m-n$ is even then $\lambda(\varphi)$ is even, and if
$m-n$ is odd then $\lambda(\varphi)$ is odd. So, if $m-n$ is even it
will be sufficient to find the elements $\varphi_{i,j,k}^s$ with
weight $0$ and if $m-n$ is odd it will be sufficient to find those
of them with weight $1$.
\end{rem}

In order to find the elements with weight $0$ or $1$, we can
consider the four sequences that correspond with the weights of
$V_1=<Y_1,Y_2,\dots,Y_{m-1},Y_m>$ (considered three times)
 and
$V_0=<X_1,X_2,\dots,X_{n-1},X_n>$:
$$
  -m+1,-m+3,\dots,m-3,m-1;
$$
$$
  -m+1,-m+3,\dots,m-3,m-1;
$$
$$
  -m+1,-m+3,\dots,m-3,m-1;
$$
$$
  -n+1,-n+3,\dots,n-3,n-1.
$$

We shall have to count the number of all possibilities to obtain
$1$ (if $m-n$ is odd) or $0$ (if $m-n$ is even). Remember that
$\lambda(\varphi_{i,j,k}^s)=\lambda(X_s)-\lambda(Y_i)-\lambda(Y_j)-\lambda(Y_k)$,
where $\lambda(X_s)$ belongs to the last sequence, and
$\lambda(Y_i)$, $\lambda(Y_j)$, $\lambda(Y_k)$ belong to the
first, second and third sequences respectively.

For example , if $m-n$ is even we have to obtain $0$, so we can
fix an element (a weight) of the last sequence and then to count
the possibilities to sum the same quantity between the three first
sequences, taking into account also the symmetry of
$\varphi_{i,j,k}^s$. In particular, we will apply this procedure
to the cases $m=3$ and odd $n$. Thus we have

\noindent {\bf Theorem 3.} {\it If $C=Z(\mu_0)\cap
\mathrm{Hom}(S^3 \gg_1,\gg_0)$ and $m=3$ and $n$ is odd, then we
have the following values for the dimension of $C$}

\

$$dim \ C=\left\{
\begin{array}{ll}
 2 & \mbox{\rm if } n=1 \\
\\
6 & \mbox{\rm if } n=3 \\ \\
8 & \mbox{\rm if } n=5 \\ \\
10 & \mbox{\rm if } n=2k+1, \ k \geq 3

\end{array}
\right.$$

\subsection{Basis}

\

In this section we are going to calculate a basis of $C$ with
$m=3$ and odd $n$. Firstly, we are going to introduce a simpler
weight of an element $\varphi \in C$. It corresponds to the action
of the diagonalizable derivation $d$, $d \in \mathrm{Der} (\mu_0)$,
 defined by:
$$
   d(X_0)=X_0,\ d(X_i)=iX_i, \ d(Y_j)=jY_j; \quad 1\leq i \leq n,\ 1 \leq j \leq m.
$$
This weight will be denoted by $p(\varphi)$. We have that
$$
    p(\varphi_{i,j,k}^s)=s-i-j-k.
$$
We have the following relationships between the two weights:
$$
   \begin{array}{l}
            \lambda(\varphi)=2p(\varphi)+3m-n+2,\\
                  p(\varphi)=\frac{1}{2}(\lambda(\varphi)-3m+n-2).
\end{array}
$$
It is clear that two $\varphi$ with different weights $p$ they
will be linearly independents.

\

Next, we are going to define some symmetric maps $\varphi$ and so we will consider $\varphi(Y_i,Y_j,Y_k)$ with $i \leq j \leq k$ and for any other reordering of the vectors, $\varphi$ will have the same value by symmetry. Thus, let $\varphi_{1,s}$ and $\varphi_{3,s}$ be elements of $\mathrm{Hom}(S^3 V_1,V_0)$
with weights $p(\varphi_{1,s})=s-3$ and $p(\varphi_{3,s})=s-5$, and defined by
$$
\begin{array}{ll}

  \varphi_{1,s}=:\left\{
\begin{array}{l} \varphi_{1,s}(Y_1,Y_1,Y_1)=X_s \\
\varphi_{1,s}(Y_1,Y_1,Y_3)=0
\end{array}
\right.
&
 \varphi_{3,s}=:\left\{
\begin{array}{l} \varphi_{3,s}(Y_1,Y_1,Y_3)=X_s \\
\varphi_{3,s}(Y_1,Y_1,Y_1)=0
\end{array}
\right.
\end{array}
$$
with $1 \leq s \leq n$ and satisfying the equations
\begin{equation}\label{ecvarphi}
  [X_0,\varphi(Y_i,Y_j,Y_k)]=
      \varphi (Y_{i+1},Y_j,Y_k)+\varphi(Y_i,Y_{j+1},Y_k)+\varphi(Y_i,Y_{j},Y_{k+1}),
      \end{equation}
\noindent with $1 \leq i \leq j \leq k \leq 3$ and $i,j \neq 3$.

\

 Thanks to the equations
(\ref{deformation}) we observe that neither $\varphi_{1,s}$ nor $\varphi_{3,s}$ are always
elements of $C$. In particular, they will be
elements of $C$ if and only if they satisfy the equations

$$
[X_0,\varphi(Y_i,Y_3,Y_3)]=
      \varphi (Y_{i+1},Y_3,Y_3), \mbox{ with } 1\leq i \leq 3.
$$

 By induction it can be proved the following formula for
$\varphi_{1,s}$ and $\varphi_{3,s}$:

$$
\begin{array}{ll}

  \varphi_{1,s}=\left\{
\begin{array}{l} \varphi_{1,s}(Y_1,Y_1,Y_1)=X_s \\
\varphi_{1,s}(Y_1,Y_1,Y_2)=\frac{1}{3}X_{s+1} \\
\varphi_{1,s}(Y_1,Y_2,Y_2)=\frac{1}{6}X_{s+2} \\
\varphi_{1,s}(Y_2,Y_2,Y_2)=\frac{1}{6}X_{s+3} \\
\varphi_{1,s}(Y_2,Y_2,Y_3)=\frac{1}{18}X_{s+4} \\
\varphi_{1,s}(Y_1,Y_3,Y_3)=-\frac{1}{18}X_{s+4} \\
\varphi_{1,s}(Y_2,Y_3,Y_3)=\frac{1}{36}X_{s+5} \\
\varphi_{1,s}(Y_3,Y_3,Y_3)=\frac{1}{36}X_{s+6} \\
\end{array}
\right.
&
 \varphi_{3,s}=\left\{
\begin{array}{l} \varphi_{3,s}(Y_1,Y_1,Y_3)=X_s \\
\varphi_{3,s}(Y_1,Y_2,Y_2)=-\frac{1}{2}X_{s} \\
\varphi_{3,s}(Y_1,Y_2,Y_3)=\frac{1}{2}X_{s+1} \\
\varphi_{3,s}(Y_2,Y_2,Y_2)=-\frac{3}{2}X_{s+1} \\
\varphi_{3,s}(Y_2,Y_2,Y_3)=-\frac{1}{2}X_{s+2} \\
\varphi_{3,s}(Y_1,Y_3,Y_3)=X_{s+2} \\
\varphi_{3,s}(Y_2,Y_3,Y_3)=-\frac{1}{4}X_{s+3} \\
\varphi_{3,s}(Y_3,Y_3,Y_3)=-\frac{1}{4}X_{s+4} \\
\end{array}
\right.
\end{array}
$$
\noindent where $\varphi_{1,s}(Y_3,Y_3,Y_3)$ and $\varphi_{3,s}(Y_3,Y_3,Y_3)$ have been completed in a natural way from  $\varphi_{1,s}(Y_2,Y_3,Y_3)$ and $\varphi_{3,s}(Y_2,Y_3,Y_3)$ respectively. Also, we supose that if $s+i > n$ then $X_{s+i}=0$.

\begin{prop} \label{varphiksbases}
The symmetric multi-linear maps $\varphi_{1,s}$ and $\varphi_{1,s}$ defined above
are elements of $C$ iff
$$
p(\varphi_{1,s})=s-3 \geq n-7 \qquad \mbox{ and } \qquad p(\varphi_{3,s})=s-5 \geq n-7
$$
\end{prop}

\begin{proof} We only have to check whether $\varphi_{1,s}$ and $\varphi_{3,s}$ satisfy or
not the equations
$$
[X_0,\varphi(Y_i,Y_3,Y_3)]=
      \varphi (Y_{i+1},Y_3,Y_3), \mbox{ with } 1\leq i \leq 3.
$$
If $p(\varphi_{1,s})= n-7$, then $s=n-4$ and $\varphi_{1,n-4}(Y_1,Y_3,Y_3)=-\frac{1}{18}X_{n}$. Thus,
$\varphi_{1,n-4}(Y_2,Y_3,Y_3)=\varphi_{1,n-4}(Y_3,Y_3,Y_3)=0$  which
clearly satisfy the above equations. If $p(\varphi_{1,s})> n-7$,
then $\varphi_{1,s}(Y_1,Y_3,Y_3)=\varphi_{1,s}(Y_2,Y_3,Y_3)=\varphi_{1,s}(Y_3,Y_3,Y_3)=0$ and
also satisfies the above equations.\\
If $p(\varphi_{1,s})< n-7$, then
$\varphi_{1,s}(Y_1,Y_3,Y_3)=-\frac{1}{18}X_{s+4}$ and
$\varphi_{1,s}(Y_2,Y_3,Y_3)=\frac{1}{36}X_{s+5}$ with $s+5 \leq n$.
But from the equation
$$[X_0, \varphi_{1,s}(Y_1,Y_3,Y_3)]=\varphi_{1,s}(Y_2,Y_3,Y_3) $$
we would have that $-\frac{1}{18}=\frac{1}{36}$ which clearly constitutes a contradiction. Analogously, it can be proved the result for $\varphi_{3,s}$.
\end{proof}

\begin{prop} Let $\varphi \in C$ with weight $p=p(\varphi)\leq
n-8$. Then
$$
  \varphi=a_1 \varphi_{1,p+3}+a_3 \varphi_{3,p+5}
$$
for some numbers $a_k$.
\end{prop}
\begin{proof}
Let $\varphi \in C$ an infinitesimal deformation with weight $p$. Then
$\varphi(Y_1,Y_1,Y_1)=a_1X_{p+3}$ and $\varphi(Y_1,Y_1,Y_3)=a_3X_{p+5}$. We are going to consider the
difference
$$
    \Psi=\varphi-a_1 \varphi_{1,p+3}+a_3 \varphi_{3,p+5}
$$
It is easy to check that $\Psi$ is a symmetric multi-linear map such
that
$$
   \Psi(Y_1,Y_1,Y_1)=\Psi(Y_1,Y_1,Y_3)=0.
$$
As $\varphi_{1,s}$ and $\varphi_{3,s}$ satisfy the equations (\ref{ecvarphi}) $\Psi$
satisfies them too, it is not difficult to see that $\Psi$ vanishes which proves the result.
\end{proof}

Without lose of generality we can consider $a_1=1$, then we have

\begin{prop} If we define $\varphi_{1,3,s}$ by $\varphi_{1,3,s}=\varphi_{1,s}+A \varphi_{3,s+2}$ and we
consider $p=s-3 \leq n-8$, then $\varphi_{1,3,s}$ is an
infinitesimal deformation of $C$ iff $A=\frac{1}{15}$ and $p \geq
n-9$.
\end{prop}

\begin{proof} As $\varphi_{1,s}$ and $\varphi_{3,s}$ satisfy the equations (\ref{ecvarphi}), $\varphi_{1,3,s}$
satisfies them too. On the other hand, by construction $\varphi_{1,s}$ and $\varphi_{3,s}$ verify too the equation
$$[X_0, \varphi(Y_2,Y_3,Y_3)]=\varphi(Y_3,Y_3,Y_3) $$
\noindent and so $\varphi_{1,3,s}$. Thus, $\varphi_{1,3,s}$ will be an
infinitesimal deformation belonging to $C$ iff it verifies the
equation

$$[X_0, \varphi_{1,3,s}(Y_1,Y_3,Y_3)]=\varphi_{1,3,s}(Y_2,Y_3,Y_3) $$

\noindent which leads to $-\frac{1}{18}+A=\frac{1}{36}-\frac{A}{4}$, that is $A=\frac{1}{15}$. Finally, as $\varphi_{1,3,s}(Y_3,Y_3,Y_3)=\frac{1}{90}X_{s+6}$
for not to obtain a contradiction it is necessary that $s+6=p+9 \geq n$. In fact, if $s+6 < n$ then $[X_0,\varphi_{1,3,s}(Y_3,Y_3,Y_3)]=0=\frac{1}{90}X_{s+7}$ which is a contradiction.
\end{proof}

\begin{rem}From now on, we will consider $A=\frac{1}{15}$ and thus
$\varphi_{1,3,s}$ would be

$$\varphi_{1,3,s}=\left\{
\begin{array}{ll} \varphi_{1,3,s}(Y_1,Y_1,Y_1)=X_s, &
\varphi_{1,3,s}(Y_1,Y_1,Y_3)=\frac{1}{15}X_{s+2} \\
\varphi_{1,3,s}(Y_1,Y_1,Y_2)=\frac{1}{3}X_{s+1}, &
\varphi_{1,3,s}(Y_1,Y_2,Y_2)=\frac{2}{15}X_{s+2} \\
\varphi_{1,3,s}(Y_2,Y_2,Y_2)=\frac{1}{15}X_{s+3}, &
\varphi_{1,3,s}(Y_1,Y_2,Y_3)=\frac{1}{30}X_{s+3} \\
\varphi_{1,3,s}(Y_2,Y_2,Y_3)=\frac{1}{45}X_{s+4}, &
\varphi_{1,3,s}(Y_1,Y_3,Y_3)=\frac{1}{90}X_{s+4} \\
\varphi_{1,3,s}(Y_2,Y_3,Y_3)=\frac{1}{90}X_{s+5}, &
\varphi_{1,3,s}(Y_3,Y_3,Y_3)=\frac{1}{90}X_{s+6} \\
\end{array}
\right.$$

\end{rem}
Thanks to the precedent results we can give a basis of C.

\

\noindent {\bf Theorem 4.} {\it If $C=Z(\mu_0)\cap
\mathrm{Hom}(S^3 \gg_1,\gg_0)$ and $m=3$ and $n$ is odd, then we
have the following vector basis of $C$}

\

\noindent \begin{itemize} \item[$\bullet$]$\{ \varphi_{1,1},
\varphi_{3,1}
\} \quad \mbox{ if } n=1$ \\

\item[$\bullet$]$\{ \varphi_{1,3}, \varphi_{1,2}, \varphi_{1,1},
\varphi_{3,3}, \varphi_{3,2}, \varphi_{3,1} \} \quad \mbox{ if }
n=3$ \\

\item[$\bullet$]$\{ \varphi_{1,5}, \varphi_{1,4}, \varphi_{1,3},
\varphi_{1,2}, \varphi_{1,1}, \varphi_{3,5}, \varphi_{3,4},
\varphi_{3,3} \} \quad \mbox{ if } n=5$ \\

\item[$\bullet$]$\{ \varphi_{1,n}, \varphi_{1,n-1},
\varphi_{1,n-2}, \varphi_{1,n-3}, \varphi_{1,n-4}, \varphi_{3,n},
\varphi_{3,n-1}, \varphi_{3,n-2}, \varphi_{1,3,n-5},
\varphi_{1,3,n-6} \} \\ \mbox{ if } n \geq 7$

\end{itemize}

\begin{rem} As we have already
 noted before, for any $\psi \in C$ we will have a filiform elementary Lie algebra of order $3$: $\mu_0 + \psi$.

\end{rem}

\section{Others families of Filiform elementary Lie algebras of order $3$}

In this section we are going to give some others families of
filiform elementary Lie algebras of order $3$  by considering
families of infinitesimal deformations $\psi$.

\

To find infinitesimal deformations we search in those
pre-infinitesimal deformations of type $C$. Thus, we consider a
family of maps, $\{ \psi \}_k$ with $k\geq 1$ that verifies:
$$
\psi_k(Y_k,Y_m,Y_m)=X_1  \mbox{ and } \psi_k(Y_i,Y_j,Y_l)=0, \
\forall \ (i,j,l) \ / \ (i,j,l)\leq (k,m,m)
$$

\noindent with $\leq$ the lexicographic order. As $\psi$ is
symmetric we consider $i\leq j \leq k$ and any other reordering
will have the same value by symmetry.

\

By applying the conditions to be  an infinitesimal deformation,
i.e.

$$[X_0, \psi_k(Y_i,Y_m,Y_m)]=\psi_k(Y_{i+1}, Y_m,Y_m), \quad k \leq i \leq m-1$$

\

\noindent we have the final expression for $\{ \psi \}_k$ with
$k\geq 1$:

$$\psi_k =\left\{ \begin{array}{ll}
\psi_k(Y_{k+i},Y_m,Y_m)=X_{1+i} & 0 \leq i \leq min\{n-1, m-k\}\\ \\
\psi_k(Y_i,Y_j,Y_l)=0 & \mbox{in any other case}
\end{array}\right.$$

Therefore we have the following Proposition

\begin{prop} The family of elementary Lie algebras of order 3, in ${\cal L}_{n,m}$, that follows
$$\{\mu_0 + \psi_k\}_{k} \quad 1 \leq k \leq m$$
 is a family of {\bf filiform} elementary Lie algebras of order 3.
\end{prop}

\begin{rem} Recall the expression of the model filiform elementary Lie algebra of order $3$,
$\mu_0$, that is the simplest filiform elementary Lie algebra. It
is defined in an adapted basis $\{X_0,X_1, \dots X_{n}, Y_1,
\dots, Y_m\}$ by the following non-null bracket products
$$\left\{\begin{array}{ll}
[X_0,X_i]=X_{i+1},& 1\leq i \leq n-1\\[1mm]
 [X_0,Y_j]=Y_{j+1},& 1\leq j
\leq m-1
\end{array}\right.$$

\

\noindent and note that in the Proposition above  we always have
$n>1$, that is, $\gg_0=<X_0,X_1, \dots X_{n}>$ is an authentic
filiform Lie algebra, because for $n \leq 1$ $\gg_0$ is an abelian
Lie algebra.
\end{rem}

Next, we present another family of infinitesimal deformations
integrable of $\mu_0$, $\psi_t$. Then $\mu_0 + \psi_t$ will be a
family of filiform elementary Lie algebras.

\begin{prop} The family of elementary Lie algebras of order 3, in ${\cal L}_{n,m}$ with $n \geq m$, that follows
$$\{\mu_0 + \psi_t \}_{t} \quad n-m \leq t \leq n$$
with

$$\psi_t= \left\{ \begin{array}{ll}
\psi_t(Y_{1},Y_1,Y_{1})=3X_{t} & \\ \\
\psi_t(Y_{1},Y_1,Y_{1+i})=X_{t+i} & 1 \leq i \leq n-t\\ \\
\psi_t(Y_i,Y_j,Y_l)=0 & \mbox{in any other case}
\end{array}\right.$$

\

\noindent  is a family of {\bf filiform} elementary Lie algebras
of order 3.
\end{prop}

\bibliographystyle{amsplain}

\end{document}